\documentclass[a4paper,11pt]{article}
\usepackage[utf8]{inputenc}
\usepackage {amssymb,latexsym,amsthm,amsmath,enumerate,amscd,epsfig,wasysym,textcomp}
\usepackage{tikz}
\usetikzlibrary{matrix}
\newtheorem{theorem}{Theorem}
\newtheorem{corollary}{Corollary}

\newtheorem{example}{Example}
\newtheorem{definition}{Definition}
\newtheorem{lemma}{Lemma}

\title{Covering Arrays on Product  Graphs}
\author{Yasmeen Akhtar$~~~~$ Soumen Maity\\
Indian Institute of Science Education and Research\\
Pune, India\\}
\date{}
\begin{document}
\maketitle

\begin{abstract}Two vectors $x,y$ in $\mathbb{Z}_g^n$ are {\it qualitatively independent} if for all pairs $(a,b)\in \mathbb{Z}_g\times \mathbb{Z}_g$, 
there exists $i\in \{1,2,\ldots,n\}$ such that $(x_i,y_i)=(a,b)$. A covering array on a graph $G$, denoted by $CA(n,G,g)$, is a $|V(G)|\times n$ array on $\mathbb{Z}_g$ 
with the property that any two rows which correspond to adjacent vertices in $G$ are qualitatively independent. The number of columns in such array is called its {\it size}. 
Given a graph $G$, a covering array on $G$ with minimum size is called {\it optimal}. Our primary concern in this paper is with constructions that make optimal covering arrays on large 
graphs those are obtained from product of smaller graphs. We consider four  most extensively studied graph products in literature and give upper and lower bounds
on the the size of covering arrays on graph products.  We find families of  graphs for which the  size of covering array on the Cartesian product  achieves the lower bound.   Finally, we present a polynomial time 
approximation algorithm with approximation ratio $\log(\frac{V}{2^{k-1}})$ for constructing covering array on 
graph $G=(V,E)$ with $k>1$ prime factors with respect to the Cartesian product.

\end{abstract}
\section{Introduction}

A covering array $CA(n,k,g)$ is a $k\times n$ array on $\mathbb{Z}_g$ with the property that any two rows are qualitatively independent. The number $n$ of columns 
in such array is called its size. The smallest possible size of a covering array is denoted 
\begin{equation*}
CAN(k,g)=\min_{n\in \mathbb{N}}\{n~:~ \mbox{there exists a } CA(n,k,g)\}
\end{equation*}
Covering arrays are generalisations of both orthogonal arrays and Sperner systems. Bounds and constructions of covering arrays have been derived from algebra, design theory, graph theory, set systems
and intersecting codes \cite{chatea, kleitman, sloane, stevens1}. Covering arrays  have industrial applications in many disparate  applications  in which factors or components interact, for example, software and circuit testing, switching networks, drug screening and data compression \cite{korner,ser,Cohen}.  In \cite{karen}, the definition of a covering array has been extended to include a graph structure. 

\begin{definition}\rm  (Covering arrays on graph).   A covering array on a graph $G$ with alphabet size $g$ and $k=|V(G)|$ is a $k\times n$ array on $\mathbb{Z}_g$. 
Each row in the array corresponds to a vertex in the graph $G$. The  array has the property that any two rows which correspond to adjacent vertices in $G$ are qualitatively independent. 
\end{definition}

\noindent A covering array on a graph $G$ will be denoted by $CA(n,G,g)$. The smallest possible covering array on a graph $G$ will be denoted
\begin{equation*}
CAN(G,g)=\min_{n\in \mathbb{N}}\{n~:~ \mbox{there exists a } CA(n,G,g)\}
\end{equation*}
Given a graph $G$ and a positive integer $g$, a covering array on $G$ with minimum size is called {\it optimal}.  Seroussi and Bshouly proved that determining the existence of an optimal binary 
covering array on a graph is an NP-complete problem \cite{ser}. We start with a review of some definitions and results from product graphs in Section \ref{productgraph}.  In Section \ref{bound}, 
we show that for all graphs $G_1$ and $G_2$, 
$$\max_{i=1,2}\{CAN(G_i,g)\}\leq CAN(G_1\Box G_2,g)\leq CAN( \max_{i=1,2}\{\chi(G_i)\},g).$$ We look for graphs $G_1$ and $G_2$ where the lower bound  on $CAN(G_1\Box G_2)$ is
achieved.   In Section \ref{Cayley}, we give families of Cayley graphs that achieves this lower bound on covering array number on graph product. In Section \ref{Approx}, we present a polynomial time 
approximation algorithm with approximation ratio $\log(\frac{V}{2^{k-1}})$ for constructing covering array on 
graph $G=(V,E)$ having more than one prime factor with respect to the Cartesian product.

\section{Preliminaries} \label{productgraph}
In this section, we give several definitions from product graphs that we use in this article. 
A graph product is a binary operation on the set of all finite graphs. However among all possible associative graph products 
the most extensively studied in literature are  the Cartesian product, the direct product,
 the strong product and  the lexicographic product. 

\begin{definition}\rm
 The Cartesian product of graphs $G$ and $H$, denoted by $G\Box H$, is the graph with 
 \begin{center}
  $V(G\Box H) = \{(g, h) \lvert g\in V(G) \mbox{ and } h \in V(H)\}$,
  \\ $E(G\Box H) = \{ (g, h)(g', h') \lvert g = g', hh' \in E(H), \mbox{ or }  gg' \in E(G), h=h' \}$.
 \end{center}
The graphs $G$ and $H$ are called the {\it factors}  of the product  $G \Box H$.
\end{definition}
\noindent In general, given  graphs $G_1,G_2,...,G_k$, then $G_1 \Box G_2 \Box \cdots \Box G_k$, is the graph with vertex set
$V(G_1) \times V(G_2) \times \cdots \times V(G_k) $, and two vertices $(x_1,x_2,\ldots, x_k)$ and
$(y_1, y_2,\ldots,y_k)$ are adjacent if and only if $x_iy_i \in E(G_i)$ for exactly one index  $1\leq i\leq k$ and $x_j = y_j$ for each index $j \not= i$.\\

\begin{definition}\rm
The direct product of graphs $G_1,G_2,...,G_k$, denoted by $G_1\times G_2\times \cdots \times G_k$, is the graph with vertex 
set $V(G_1) \times V(G_2) \times \cdots \times V(G_k) $, and for which vertices $(x_1,x_2,...,x_k)$ and $(y_1,y_2,...,y_k)$ are 
adjacent precisely if  $x_iy_i \in E(G_i)$ for each index $i$. 
\end{definition}

\begin{definition}\rm
The strong product of graphs $G_1,G_2,...,G_k$, denoted by $G_1\boxtimes G_2\boxtimes \cdots \boxtimes G_k$, is the graph with vertex set 
$V(G_1) \times V(G_2) \times \cdots \times V(G_k) $, and distinct vertices $(x_1,x_2,\ldots,x_k)$ and $(y_1,y_2,\ldots,y_k)$  are adjacent if and only if 
either $x_iy_i\in E(G_i)$ or $x_i=y_i$ for each $1\leq i\leq k$. We note that in general  $E(\boxtimes_{i=1}^k {G_i}) \neq E(\Box_{i=1}^k G_i) \cup E(\times_{i=1}^k G_i)$, unless $k=2$.
\end{definition}

\begin{definition}\rm
 The lexicographic product of graphs $G_1,G_2,...,G_k$, denoted by $G_1\circ G_2\circ \cdots \circ G_k$, is the graph with 
 vertex set  $V(G_1) \times V(G_2) \times \cdots \times V(G_k) $, and two vertices $(x_1,x_2,...,x_k)$ and $(y_1,y_2,...,y_k)$ are 
adjacent if and only if for some index $j\in \{1,2,...,k\}$ we have $x_jy_j \in E(G_j)$ and $x_i =y_i$ for each index $1\leq i < j$. 
\end{definition}

Let $G$ and $H$ be graphs with vertex sets $V(G)$ and $V(H)$, respectively. A {\it homomorphism} from $G$ to $H$ is  a map 
$\varphi~:~V(G)\rightarrow V(H)$ that preserves adjacency:  if $uv$ is an edge in $G$, then $\varphi(u)\varphi(v)$ is an edge in $H$. 
We say $G\rightarrow H$ if there is a homomorphism from $G$ to $H$,  and $G \equiv H$  if $G\rightarrow H$ and $H\rightarrow G$. 
A {\it weak homomorphism} from $G$ to $H$ is a map  $\varphi~:~V(G)\rightarrow V(H)$ such that if $uv$ is an edge in $G$, then either 
$\varphi(u)\varphi(v)$ is an edge in $H$, or $\varphi(u)=\varphi(v)$. Clearly every homomorphism is automatically a weak homomorphism.

Let $\ast$  represent either the  Cartesian, the direct or the strong product of graphs, and consider a product $G_1\ast G_2\ast \ldots\ast G_k$. 
For any index $i$, $1\leq i\leq k$, a {\it projection map} is defined as:
$$p_i~:~G_1\ast G_2\ast \ldots\ast G_k \rightarrow G_i ~\mbox{where} ~p_i(x_1,x_2,\ldots,x_k)=x_i.$$  By the definition of the Cartesian, the direct, and the strong product of 
graphs, each $p_i$ is a weak homomorphism. In the case of direct product, as $(x_1,x_2,\ldots,x_k)(y_1,y_2,\ldots,y_k)$ is an an edge of $G_1\times G_2 \times,\ldots,\times G_k$ 
if and only if $x_iy_i\in E(G_i)$ for each $1\leq i\leq k$., each projection $p_i$ is actually a homomorphism. In the case of lexicographic product, the first projection map that is projection on first component is a weak homomorphism, where as in general the projections to the other 
components are not weak homomorphisms. \\

A graph is {\it prime} with respect to a given graph product if it is nontrivial and cannot be represented as the product of two nontrivial 
graphs. For the Cartesian product,
it means that a nontrivial graph $G$ is prime if $G=G_1\Box G_2$ implies that either $G_1$ or $G_2$ is $K_1$. Similar observation is 
true for other three products. The uniqueness of the prime factor decomposition of connected graphs with respect to the
 Cartesian product was first shown by Subidussi $(1960)$, and independently by Vizing $(1963)$. Prime factorization is not unique 
 for the Cartesian product in the class of possibly disconnected simple graphs \cite{HBGP}. It is known that any connected graph factors 
 uniquely into prime graphs with respect to the Cartesian product. 
 
 \begin{theorem}(Sabidussi-Vizing)
Every connected graph has a unique representation as a product of prime graphs, up to isomorphism and the order of the factors. The number of prime factors is 
at most  $\log_2 {V}$.
 \end{theorem}
\noindent For any connected graph $G=(V,E)$, the prime factors of $G$  with respect to the Cartesian product can be computed in $O(E \log V) $ times and $O(E)$ space. See Chapter 23,  \cite{HBGP}.

\section{Graph products and covering arrays}\label{bound}
Let $\ast$ represent either the Cartesian, the direct, the strong, or the lexicographic product operation. 
Given covering arrays $CA(n_1,G_1,g)$ and $CA(n_2,G_2,g)$, one can construct covering array on $G_1 \ast G_2$ as follows:  the  row corresponds
 to the vertex $(a,b)$ is obtained by horizontally concatenating the row corresponds to the vertex $a$ in $CA(n_1,G_1,g)$  with the row
 corresponds to the vertex $b$ in $CA(n_2,G_2,g)$. Hence an obvious upper bound for the covering array number is given by
 \begin{center}
  $CAN(G_1 \ast G_2, g) \leq CAN(G_1, g) + CAN(G_2, g) $
 \end{center}
  We now propose some  improvements of this bound. A column of a covering array is {\it constant} if,  for some symbol $v$, every entry in the 
  column is $v$. In a {\it standardized } $CA(n,G,g)$ the first column is constant. Because symbols within each row can be permuted independently, 
  if a $CA(n,G,g)$ exists, then a standardized $CA(n,G,g)$ exists. 
\begin{theorem}
 Let $G=G_1\boxtimes G_2\boxtimes \cdots \boxtimes G_k$, $k\geq 2$ and $g$ be a positive integer. 
 Suppose for each $1\leq i\leq k$  there exists a $CA(n_i,G_i,g)$, then there exists a 
 $CA(n,G,g)$ where $n=\underset{i=1}{\overset{k}\sum} n_i -k$. Hence,
 $CAN(G,g)\leq \underset{i=1}{\overset{k}\sum} CAN(G_i,g)-k$.
 
\end{theorem}

\begin{proof} Without loss of generality, we assume that  for each $1\leq i\leq g$, the first column of $CA(n_i,G_i,g)$
 is a constant column on symbol $i$ and for each $g+1\leq i\leq k$, the first column of $CA(n_i,G_i,g)$ is a constant 
 column on symbol 1. 
  Let $C_i$ be the array 
 obtained from $CA(n_i,G_i,g)$ by removing the first column. Form an array $A$ with 
 $\underset{i=1}{\overset{k}\prod} |V(G_i)|$ rows and 
 $\underset{i=1}{\overset{k}\sum} n_i -k$ columns, indexing rows as $(v_1,v_2,...,v_k)$, where $v_i\in V(G_i)$.
 Row $(v_1,v_2,...,v_k)$ is 
 obtained  by horizontally concatenating the rows correspond to the vertex $v_i$ of  $C_i$, for $1\leq i\leq k$. 
 Consider two distinct rows   $(u_1,u_2,\ldots,u_k)$ and $(v_1,v_2,\ldots,v_k)$  of $A$ which correspond to adjacent vertices in $G$.  
  Two distinct vertices $(u_1,u_2,\ldots,u_k)$ and $(v_1,v_2,\ldots,v_k)$  are adjacent if and only if 
either $u_iv_i\in E(G_i)$ or $u_i=v_i$ for each $1\leq i\leq k$. Since the vertices are distinct, $u_iv_i\in E(G_i)$ for at least one index $i$.
When $u_i=v_i$, all pairs of the form $(a,a)$ are covered. When $u_iv_i\in E(G_i)$ all remaining pairs are covered because two different rows of $C_i$ correspond to adjacent vertices in $G_i$ are selected.

\end{proof}

\noindent Using the definition of strong product of graphs we have following result as a corollary.
\begin{corollary}
 Let $G=G_1\ast G_2\ast \cdots \ast G_k$, $k\geq 2$ and $g$ be a positive integer, where $\ast\in\{\Box,\times\}$. Then,
  $CAN(G,g)\leq \underset{i=1}{\overset{k}\sum} CAN(G_i,g)-k$.
 
\end{corollary}

\noindent The lemma given below will be used in Theorem \ref{product}. 

 \begin{lemma}\label{karenlemma} (Meagher  and Stevens \cite{karen})
 Let $G$ and $H$ be graphs. If $G\rightarrow H$  then $CAN(G,g)\leq CAN(H,g)$.
 
 \end{lemma}

\begin{theorem}\label{product}
  Let $G=G_1\times G_2\times \cdots \times G_k$, $k\geq 2$ and $g$ be a positive integer. 
 Suppose for each $1\leq i\leq k$ there exists a $CA(n_i,G_i,g)$.  Then there exists a 
 $CA(n,G,g)$ where $n=\min\limits_{i} n_i$. Hence, $CAN(G,g)\leq \underset{i}{\overset{}\min}$ $ CAN(G_i,g)$.
 
\end{theorem}
\begin{proof}
 Without loss of generality assume that $n_1 = \min\limits_{i} {n_i} $. It is known  that $G_1\times G_2\times \cdots \times G_k\rightarrow G_1$.   Using Lemma \ref{karenlemma}, we have $CAN(G,g)\leq CAN(G_1,g)$.

\end{proof}

\begin{theorem}
 Let $G=G_1\circ G_2\circ \cdots \circ G_k$, $k\geq 2$ and $g$ be a positive integer. 
 Suppose for each $1\leq i\leq k$  there exists a $CA(n_i,G_i,g)$. Then there exists a 
 $CA(n,G,g)$ where $n=\underset{i=1}{\overset{k}\sum} n_i -k+1$. Hence, 
  $CAN(G,g)\leq \underset{i=1}{\overset{k}\sum} CAN(G_i,g)-k+1$.
 \end{theorem}
\begin{proof} We assume that  for each $1\leq i\leq k$, the first column of $CA(n_i,G_i,g)$
 is a constant column on symbol $1$.
 Let $C_1= CA(n_1,G_1,g)$.  For each $2\leq i\leq k$ remove the first column of $CA(n_i,G_i,g)$ to form $C_i$ with $n_i-1$ columns. Without loss of generality assume first column of each $CA(n_i,G_i,g)$ is constant 
 vector on symbol 1 while for each $2\leq i\leq k$, $C_i$ is the array obtained from $CA(n_i,G_i,g)$ by removing the first 
 column. Form an array $A$ with $\underset{i=1}{\overset{k}\prod} |V(G_i)|$  rows and  $\underset{i=1}{\overset{k}\sum} n_i -k+1$ columns, indexing 
 rows as   $(v_1,v_2,..,v_k)$, $v_i\in V(G_i)$.  Row $(v_1,v_2,\ldots,v_k)$ is  obtained  by horizontally
 concatenating the rows correspond to the vertex $v_i$ of $C_i$, for $1\leq i\leq k$. If two vertices 
 $(v_1,v_2,...,v_k)$ and $(u_1,u_2,...,u_k)$ are adjacent in $G$ then either $v_1u_1\in E(G_1)$ or $v_ju_j\in E(G_j)$ for 
 some $j\geq 2$ and $v_i=u_i$ for each $i< j$. In first case rows from $C_1$ covers each ordered pair of symbols while in second case 
 rows from $C_j$ covers each ordered pair of symbol probably except $(1,1)$. But this pair appears in each $C_i$ for $i<j$. Hence $A$ 
 is a covering array on $G$.
\end{proof}

\begin{definition} \rm A {\it proper colouring} on a graph is an assignment of colours to each vertex such that adjacent vertices receive a different colour. The chromatic number of a graph $G$, $\chi(G)$, 
is defined to be the size of the smallest set of colours such that a proper colouring exists with that set.
\end{definition} 
 
 \begin{definition}\rm 
 A {\it  maximum clique} in a graph $G$ is a maximum set of pairwise adjacent vertices. The maximum clique number of a graph $G$, $\omega(G)$, is defined to be the size of a maximum clique. 
 \end{definition}
 
\noindent Since there are homomorphisms $K_{\omega(G)}\rightarrow G\rightarrow K_{\chi(G)}$, we can 
 find bound on the size of a covering array on a graph from the graph's chromatic number and clique number.  For all graphs $G$,
$$CAN(K_{\omega(G)},g)\leq CAN(G,g)\leq CAN(K_{\chi(G)},g).$$

\noindent  We have  the following  results on proper colouring of product graphs  \cite{chromatic} 
 $$\chi(G_1 \Box G_2) = \max \{ \chi(G_1), \chi(G_2)\}.$$ 
 For other graph products there are no explicit formulae for chromatic number but following bounds are mentioned in \cite{HBGP}.
 $$\chi(G_1 \times G_2) \leq \min \{ \chi(G_1), \chi(G_2)\}$$
 $$\chi(G_1 \boxtimes G_2) \leq \chi(G_1 \circ G_2) \leq  \chi(G_1) \chi(G_2).$$
  A proper colouring of $G_1 \ast G_2$ with $\chi(G_1 \ast G_2)$ colours is equivalent to a  homomorphism from 
  $G_1 \ast G_2$ to $K_{\chi(G_1 \ast G_2)}$ for any $\ast \in\{\Box, \times, \boxtimes, \circ \}$. 
 Hence $$CAN(G_1 \Box G_2, g) \leq CAN(K_{\max\{ \chi(G_1), \chi(G_2)\}},g)$$
 $$CAN(G_1 \times G_2, g) \leq CAN(K_{\min\{ \chi(G_1), \chi(G_2)\}},g) $$
 $$CAN(G_1 \boxtimes G_2, g) \leq CAN(K_{\chi(G_1)\chi(G_2)},g) $$
 $$CAN(G_1 \circ G_2, g) \leq CAN(K_{\chi(G_1)\chi(G_2)},g) .$$
\noindent  Note that $G_1\rightarrow G_1 \ast G_2$ and $G_2\rightarrow G_1 \ast G_2$  for $\ast \in\{\Box,\boxtimes,\circ\}$ 
which gives 
  $$max\{CAN(G_1, g), CAN(G_2, g)\}\leq CAN(G_1 \ast G_2, g).$$
 We now describe colouring construction of covering array on graph $G$. If $G$ is a $k$-colourable graph then build a covering array $CA(n, k, g)$ and without loss of generality associate 
 row $i$ of  $CA(n, k, g)$ with colour $i$ for $1\leq i\leq k$.  In order to construct $CA(n,G,g)$, we assign row $i$ of $CA(n, k, g)$ to all the vertices having colour $i$ in $G$. 
 
 \begin{definition}\rm An orthogonal array $OA(k,g)$ is a $k\times g^2$ array with entries from $\mathbb{Z}_g$ having the properties that 
 in every two rows, each ordered pair of symbols from $\mathbb{Z}_g$ occurs exactly once.  
 \end{definition}
 
 \begin{theorem}\label{OA} \cite{Colbourn} If $g$ is prime or power of prime, then one can construct $OA(g+1,g)$.  
 \end{theorem} 
 The set of rows in an orthogonal array $OA(k,g)$ is a set of $k$ pairwise qualitatively independent vectors from 
 $\mathbb{Z}_g^{g^2}$. For $g=2$, by Theorem \ref{OA},  there are three qualitatively independent vectors from 
 $\mathbb{Z}_2^{4}$.  Here we give some examples where the lower bound
 on $CAN(G_1\Box G_2,g)$ is achieved, that is, $CAN(G_1\Box G_2,g)=max\{CAN(G_1,g), CAN(G_2,g)\}. $ 
 
\begin{example} \rm If  $G_1$ and $G_2$ are bicolorable graphs, then $\chi(G_1 \Box G_2)=2$. Let $x_1$ and $x_2$  be two qualitatively independent vectors
in $\mathbb{Z}_g^{g^2}$. Assign vector $x_i$ to all the vertices of $G_1 \Box G_2$ having colour $i$ for $i=1,2$ to get a covering array with $CAN(G_1 \Box G_2, g) = g^2.$
 \end{example}

\begin{example}\rm  If  $G_1$ and $G_2$ are complete graphs,  then  $CAN(G_1 \Box G_2, g) = max\{CAN(G_1, g), CAN(G_2, g)\}. $
 \end{example}
  
  \begin{example} \rm  If $G_1$ is bicolorable and $G_2$ is a complete graph on $k\geq 3$ vertices, then 
  $CAN(G_1 \Box G_2, g) = CAN(G_2, g)$. In general,  if $\chi(G_1) \leq \chi(G_2)$ and $G_2$ is a complete graph, then
  $CAN(G_1 \Box G_2, g) = CAN(G_2, g)$.
 \end{example}

\begin{example} \rm  If $P_m$ is a path of length $m$ and $C_n$ is an odd cycle of length $n$, then $\chi(P_m \Box C_n)=3$. Using Theorem \ref{OA}, we 
get a set  of  three qualitatively independent vectors in  $\mathbb{Z}_g^{g^2}$ for $g\geq 2$.  Then the colouring construction of covering arrays gives us a covering 
array on $P_m\Box C_n$ with $CAN(P_m\Box C_n, g) = g^2$. 
 \end{example}
 
\begin{lemma}\cite{HBGP} Let $G_1$ and $G_2$ be graphs and $Q$ be a clique of $G_1\boxtimes G_2$. Then 
 $Q= p_1(Q)\boxtimes p_2(Q)$, where $p_1(Q)$ and $p_2(Q)$ are cliques of $G_1$  and $G_2$, respectively. 
 \end{lemma}
Hence a maximum size clique of $G_1\boxtimes G_2$ is product of maximum size cliques from $G_1$ and $G_2$. That is, 
$\omega(G_1\boxtimes G_2)= \omega(G_1)\omega(G_2)$. Using the graph homomorphism, this results into another lower bound on 
$CAN(G_1\boxtimes G_2,g)$ as $CAN(K_{\omega(G_1)\omega(G_2)},g)\leq CAN(G_1\boxtimes G_2,g)$. Following are some examples 
where this lower bound can be achieved.
\begin{example} \rm If $G_1$ and $G_2$ are nontrivial bipartite graphs then 
 $\omega(G_1 \boxtimes G_2)= \chi(G_1\boxtimes G_2)$ which is 4. Hence $CAN(G_1\boxtimes G_2,g)= CAN(K_4,g)$, which is of optimal
 size.
\end{example}

\begin{example}\rm If  $G_1$ and $G_2$ are complete graphs,  then $G_1\boxtimes G_2$ is again a complete graph. Hence 
$CAN(G_1\boxtimes G_2,g)= CAN(K_{\omega(G_1\boxtimes G_2)},g)$.
\end{example}

\begin{example}\rm If $G_1$ is a bipartite graph and $G_2$ is a complete graph on $k\geq 2$ vertices, then 
$\omega(G_1\boxtimes G_2)= \chi(G_1\boxtimes G_2)= 2k$. Hence  $CAN(G_1\boxtimes G_2,g)= CAN(K_{2k},g)$.
\end{example}

\begin{example}\rm If $P_m$ is a path of length $m$ and $C_n$ is an odd cycle of length $n$, then 
$\omega(P_m\boxtimes C_n)=4$ and $\chi(P_m \boxtimes C_n)=5$.  Here we have $CAN(K_4,g)\leq CAN(G,g)\leq CAN(K_5,g)$.  For $g\geq 4$, using Theorem \ref{OA}, we 
get a set  of  five qualitatively independent vectors in  $\mathbb{Z}_g^{g^2}$.  Then the colouring construction of 
covering arrays gives us a covering array on $P_m\boxtimes C_n$ with $CAN(P_m\boxtimes C_n, g) = g^2$. 
 
\end{example}

 \section{Optimal size covering arrays over the  Cartesian product of  graphs } \label{Cayley}
 
 \begin{definition} \rm Two graphs $G_1=(V,E)$ and $G_2=(V^{\prime},E^{\prime})$ are said to be isomorphic if there is a bijection mapping $\varphi$ from the vertex set $V$ to the vertex set $V^{\prime}$ such that $(u,v)\in E$ if and only if $(\varphi(u),\varphi(v))\in E^{\prime}$. The mapping $\varphi$ is called an isomorphism. An automorphism of a graph is an isomorphism from the graph to itself. 
 \end{definition}
 \noindent  The set of all automorphisms of a graph $G$ forms a group, denoted $Aut(G)$, the automorphism group of $G$.

 \begin{theorem}\label{A}
Let $G_1$ be a graph having the property that $Aut(G_1)$ contains a fixed point free automorphism which maps every vertex to its neighbour. 
Then for any bicolourable graph $G_2$, $$CAN(G_1 \square  G_2,g)=CAN(G_1,g).$$
\end{theorem}

\begin{proof}  Consider the set $\Gamma=\{\phi \in Aut(G_1)~|~ \phi(u)\in N(u)-\{u\} \mbox{ for all } u\in V(G_1)\}$ where $N(u)$ denotes the set of neighbours of $u$. 
 From the assumption, $\Gamma$ is not empty. 
Consider a 2-colouring of $G_2$ with colours $0$ and $1$.  Let $W_0=\{(u,v)\in V(G_1\square G_2) ~|~\mbox{colour}(v)=0\}$ and $W_1=\{(u,v)\in V(G_1\square G_2) ~|~\mbox{colour}(v)=1\}$. Note that $W_0 $ and $W_1$ partition $V(G_1\square G_2)$ in two two parts. 
 Let the rows of covering array $CA(G_1,g)$ be indexed by $u_1,u_2,\ldots,u_k$.
 Form an array $C$ with $|V(G_1 \Box G_2)|$ rows and $CAN(G_1,g)$
 columns, indexing rows as $(u,v)$ for $1\leq u\leq |V(G_1)|$, $1\leq v \leq  |V(G_2)|$. If $(u,v)\in W_0$, row $(u,v)$ is row $u$ of $CA(G_1,g)$; otherwise if 
  $(u,v)\in W_1$, row $(u,v)$ is row $\phi(u)$ of $CA(G_1,g)$.   We verify that $C$ is a $CA(G_1\Box G_2, g)$. Consider two adjacent vertices $(u_1,v_1)$ and $(u_2,v_2)$ 
  of $C$. \\ (i) Let $(u_1,v_1)$ and $(u_2,v_2)$ belong to $W_i$, then $(u_1,v_1)\sim(u_2,v_2)$ if and only if $u_1 \sim u_2$ and $v_1=v_2$. 
  When $(u_1,v_1)$ and $(u_2,v_2)$ belong to $W_0$, rows $(u_1,v_1)$ and $(u_2,v_2)$ are rows $u_1$ and $u_2$ of $CA(G_1,g)$ respectively.  
  As $u_1\sim u_2$,  rows  $u_1$ and $u_2$ are 
qualitatively independent in $CA(G_1,g)$. When $(u_1,v_1)$ and $(u_2,v_2)$ belong to $W_1$, rows $(u_1,v_1)$ and $(u_2,v_2)$ are rows $\phi(u_1)$ and $\phi(u_2)$ of $CA(G_1,g)$ respectively.  As $\phi(u_1)\sim \phi(u_2)$,  rows  $\phi(u_1)$ and $\phi(u_2)$ are 
qualitatively independent in $CA(G_1,g)$. Therefore, rows  $(u_1,v_1)$ and $(u_2,v_2)$ are 
qualitatively independent in $C$.\\
(ii) Let  $(u_1,v_1)\in W_0$ and $(u_2,v_2)\in W_1$.    In this case, $ (u_1,v_1) \sim (u_2,v_2)$ if and only if  $u_1=u_2$ and $v_1\sim v_2$.  Let $u_1=u_2=u$.  Rows $(u,v_1)$  and  $(u,v_2)$ are rows $u$ and $\phi(u)$  of  $CA(G_1,g)$. 
 As $\phi $ is a fixed point free automorphism that maps every vertex to its neighbour, $u$ and $\phi(u)$ are adjacent in $G_1$. Therefore, the rows indexed by $u$ and $\phi(u)$ are qualitatively independent 
in $CA(G_1,g)$; therefore, rows  $(u_1,v_1)$ and $(u_2,v_2)$ are 
qualitatively independent in $C$.\\
\end{proof}

\begin{definition}\rm
 Let $H$ be a finite group and $S$ be a subset of $H\smallsetminus \{id\}$ such that $S = -S$ (i.e., $S$ is closed under inverse). The Cayley graph of $H$ generated by $S$, denoted 
  $Cay(H,S)$, is the undirected graph $G=(V,E)$ where $V=H$ and $E=\{(x,sx)~|~x\in H, s\in S\}$.   The Cayley graph is connected if and only if $S$ generates  $H$.
\end{definition}
\noindent Through out this article by $S = -S$ we mean, $S$ is closed under inverse for a given group operation

\begin{definition}\rm
 A  circulant graph $G(n,S)$ is a Cayley graph  on $\mathbb{Z}_n$.  That is, it is a graph whose vertices are labelled $\{0,1,\ldots,n-1\}$, with two vertices labelled $i$ and
 $j$  adjacent iff $i-j ~(\mbox{mod}~n)\in S$, where $S\subset \mathbb{Z}_n$ with $S=-S$ and $0\notin S$. 
 \end{definition}
\begin{corollary}
 Let $G_1(n,S)$ be a circulant graph and $G_2$ be a bicolorable graph, then $CAN(G_1(n,S) \Box G_2, g) = CAN(G_1(n,S), g)$.
 
\end{corollary}
\begin{proof} Let $i$ and $j$ be any two adjacent vertices in $G_1(n,S)$.  We define a mapping $\phi$ from  $\mathbb{Z}_n$ as follows:
 \begin{center}
  $\phi(k) = k+j-i ~(\mbox{mod}~ n)$
 \end{center}
 It is easy to verify that $\phi$ is an automorphism 
 and it sends every vertex to its neighbour. Hence $\phi \in \Gamma$ and the result
  follows.
\end{proof}

For a group $H$ and $S \subseteq H$, we denote conjugation of $S$ by elements of itself as
\begin{center}
 $S^S = \{ ss's^{-1}  | s, s'\in S\}$
\end{center}

\begin{corollary}
 Let $H$ be a finite group and $S \subseteq H\smallsetminus \{id\}$ is a generating set for $H$ such that $S = -S$ and 
 $S^S = S$. Then for
 $G_1 = Cay(H, S)$ and any bicolorable graph $G_2$,
 \begin{center}
  $CAN(G_1 \Box G_2, g) = CAN(G_1, g)$
 \end{center}
\end{corollary}
\begin{proof}
 We will show that there exists a $\phi \in Aut(G_1)$ such that $\phi$ is stabilizer free. 
 Define $\phi : H \rightarrow H$ as $\phi(h) = sh$ for some $s\in S$.
 It it easy to check that $\phi$ is bijective and being $s \neq id$ it is stabilizer free. Now to prove it is a 
 graph homomorphism we need to show it is an adjacency preserving map. It is sufficient to prove that $(h, s'h)\in E(G_1)$ 
  implies $(sh, ss'h) \in E(G_1)$.  As $ss'h = ss's^{-1}sh$ and $ss's^{-1} \in S$, we have $(sh, ss'h)\in E(G_1)$. 
 Hence $\phi \in \Gamma $ and Theorem \ref{A} implies the result.
\end{proof}
\begin{example}\rm
 For any abelian group $H$ and $S$ be a generating set such that $S = -S$ and $id \notin S$, we always get $S^S = S$.
\end{example}

\begin{example}\rm
For  $H = Q_8 = \{\pm1, \pm i, \pm j, \pm k\}$ and $S = \{\pm i, \pm j\}$, we have $S^S = S$ and $S = -S$.
\end{example}
\begin{example}\rm 
For  $H= D_8 = \langle a, b | a^2 = 1 = b^4, aba = b^3\rangle$ and $S= \{ab, ba\}$, we have $S^S = S$ and $S = -S$.
\end{example}
\begin{example}\rm 
For  $H = S_n$ and $S =$ set of all even cycles, we have $S^S = S$ and $S = -S$
\end{example}

\begin{theorem}
 Let $H$ be a finite group and $S$ be a generating set for $H$ such that
 \begin{enumerate}
  \item $S = -S$ and $id \notin S$
  \item $S^S = S$
  \item  there exist $s_1$ and $s_2$ in $S$ such that $s_1 \neq s_2$ and $s_1s_2 \in S$
 \end{enumerate}
then for $G_1= Cay(H, S)$ and any three colourable graph $G_2$
\begin{center}
 $CAN(G_1 \Box G_2, g) = CAN(G_1,g)$
\end{center}
\end{theorem}
\begin{proof} Define three distinct automorphisms of $G_1$,  $\sigma_{i} : H\rightarrow H$, for $i=0,1,2$, as $\sigma_0(u)=u$, $\sigma_1(u)=s_1u$, $\sigma_2(u)=s_2^{-1}u$. 
 Consider a three colouring of $G_2$ using the colours $0, 1$ and $2$.  Let $W_i=\{(u,v)\in V(G_1\square G_2) ~|~\mbox{colour}(v)=i\}$  for $i=0,1,2$. 
  Note that $W_0 $, $W_1$, and $W_2$ partition $V(G_1\square G_2)$ into three parts. 
 Let the rows of covering array $CA(G_1,g)$ be indexed by $u_1,u_2,\ldots,u_k$. Using $CA(G_1,g)$, form an array $C$ with $|V(G_1 \Box G_2)|$ rows and $CAN(G_1,g)$
 columns, indexing rows as $(u,v)$ for $1\leq u\leq |V(G_1)|$, $1\leq v \leq  |V(G_2)|$. If $(u,v)\in W_i$, row $(u,v)$ is row $\sigma_i(u)$ of $CA(G_1,g)$.  Consider two adjacent vertices  $(u_1,v_1)$ and  $(u_2,v_2)$ of $C$. \\
 (i) Let $(u_1,v_1)$ and $(u_2,v_2)$ belong to $W_i$. In this case,  $(u_1,v_1)\sim(u_2,v_2)$ if and only if $u_1 \sim u_2$ and $v_1=v_2$.  
 When $(u_1,v_1)$ and $(u_2,v_2)$ belong to $W_0$, rows $(u_1,v_1)$ and $(u_2,v_2)$ are rows $u_1$ and $u_2$ of $CA(G_1,g)$.  
 As $u_1 \sim u_2$ in $G_1$, the rows  $u_1$ and $u_2$ are qualitatively independent in $CA(G_1,g)$.    Let $(u_1,v_1)$ and $(u_2,v_2)$ belong to $W_1$ (res. $W_2$).   Similarly, as $s_1u_1\sim s_1u_2$ (res. $s_2^{-1}u_1 \sim s_2^{-1}u_1$)
the rows indexed by  
$s_1u_1$ and $s_1u_2$  (res. $s_2^{-1}u_1$ and $s_2^{-1}u_2$) are qualitatively independent in $CA(G_1,g)$.
Hence the rows  
$(u_1,v_1)$ and $(u_2,v_2)$  are qualitatively independent in $C$.\\ 
(ii) Let $(u_1,v_1)\in W_i$ and $(u_2,v_2)\in W_j$ for $0\leq i\neq j\leq 2$. In this case, $(u_1,v_1)\sim(u_2,v_2)$ if and only if $u_1 = u_2$ and $v_1\sim v_2$. 
Let $u_1=u_2=u$.\\ 
Let $(u,v_1)\in W_0$ and $(u,v_2)\in W_1$, then rows $(u,v_1)$ and $(u,v_2)$ are  rows $u$ and $s_1u$ of $CA(G_1,g)$ respectively.   Then as $u\sim s_1u$ the rows indexed by $(u,v_1)\in W_0$ and $(u,v_2)\in W_1$ are qualitatively independent in $C$. \\
Let $(u,v_1)\in W_0$ and $(u,v_2)\in W_2$.  Then, as $u\sim s_2^{-1}u$, the rows indexed by $(u,v_1)\in W_0$ and $(u,v_2)\in W_2$ are qualitatively independent in $C$. \\
Let $(u,v_1)\in W_1$ and $(u,v_2)\in W_2$.  Then, as $s_1u\sim s_2^{-1}u$, the rows indexed by $(u,v_1)\in W_1$ and $(u,v_2)\in W_2$ are qualitatively independent in $C$. 
\end{proof}

 \begin{theorem}
 Let $H$ be a finite group and $S$ is a generating set for $H$ such that
 \begin{enumerate}
  \item $S = -S$ and $id \notin S$
  \item $S^S = S$
  \item $\exists s_1$ and $s_2$ in $S$ such that $s_1 \neq s_2$ and $s_1s_2,  s_1s_2^{-1}\in S$
 \end{enumerate}
then for $G_1 = Cay(H, S)$ and any four colourable graph $G_2$
\begin{center}
 $CAN(G_1 \Box G_2, g) = CAN(G_1, g)$
\end{center}
\end{theorem}
\begin{proof}
Define four distinct automorphisms of $G_1$, $\sigma_i:H\rightarrow H$, $ i=0,1,2,3$ as $\sigma_0(u)=u$, $\sigma_1(u)=s_1u$, $\sigma_2(u)=s_2u$ and 
$\sigma_3(u)=s_1s_2 u$. Consider a four colouring of $G_2$ using the colours $0, 1, 2$ and $3$.  Let $W_i=\{(u,v)\in V(G_1\square G_2) ~|~\mbox{colour}(v)=i\}$  for $i=0,1,2,3$. 
 Let the rows of covering array $CA(G_1,g)$ be indexed by $u_1,u_2,\ldots,u_k$.  Form an array $C$ with $|V(G_1 \Box G_2)|$ rows and $CAN(G_1,g)$
 columns, indexing rows as $(u,v)$ for $1\leq u\leq |V(G_1)|$, $1\leq v \leq  |V(G_2)|$. If $(u,v)\in W_i$, row $(u,v)$ is row $\sigma_i(u)$ of $CA(G_1,g)$.  Consider two adjacent vertices  $(u_1,v_1)$ and  $(u_2,v_2)$ of $C$. \\
 (i) Let $(u_1,v_1)$ and $(u_2,v_2)$ belong to $W_i$. It is easy to verify that $(u_1,v_1)$ and $(u_2,v_2)$  are qualitatively independent.\\
(ii) Let $(u_1,v_1)\in W_i$ and $(u_2,v_2)\in W_j$ for $0 \leq i\neq j\leq 3$.   In this case, $(u_1,v_1)\sim(u_2,v_2)$ if and only if $u_1 = u_2$ and $v_1\sim v_2$. 
Let $u_1=u_2=u$.\\ 
Let $(u,v_1)\in W_0$ and $(u,v_2)\in W_i$ for $i=1,2,3$, then  row  $(u,v_1)$  and  $(u,v_2)$ are
 rows $u$ and  $\sigma_i(u)$ of $CA(G_1,g)$ respectively. 
Then as $u\sim \sigma_i(u)$ the rows $(u,v_1)$ and $(u,v_2)$ are qualitatively independent. \\

\noindent Let $(u,v_1)\in W_1$ and $(u,v_2)\in W_2$. Then rows $(u,v_1)$ and $(u,v_2)$ are rows $s_1u$ and $s_2u$ of $CA(G_1,g)$.   As $s_1u = s_1s_2^{-1}s_2u$ and $s_1s_2^{-1}\in S$, we get $s_1u\sim s_2u$. Hence the rows $(u,v_1)\in W_1$ and $(u,v_2)\in W_2$ are qualitatively independent. Similarly, as $s_1u=s_1 s_2^{-1}s_1^{-1}s_1s_2u$ and $s_1 s_2^{-1}s_1^{-1}\in S$ being $S^S=S$, we have 
$s_1u\sim s_1s_2u$. Hence the rows $(u,v_1)\in W_1$ and $(u,v_2)\in W_3$ are qualitatively independent. \\
Let $(u,v_1)\in W_2$ and $(u,v_2)\in W_3$. As $s_2u=s_1^{-1}s_1s_2u$ and $s_1^{-1}\in S$, we get $s_2u\sim s_1s_2u$. 
Hence the rows $(u,v_1)\in W_2$ and $(u,v_2)\in W_3$ are qualitatively independent. 
\end{proof}

\begin{example}
 $G = Q_8$ and $S= \{\pm i, \pm j, \pm k\}$. Here $s_1=i$ and $s_2=j$. 
\end{example}

\begin{example}
 $G = Q_8$ and $S= \{-1,\pm i, \pm j\}$. Here $s_1=-1$ and $s_2=i$. 
\end{example}
\begin{figure}
\begin{center}
\begin{tikzpicture}
\small{
\matrix[matrix of math nodes, anchor=south west,
        nodes={circle, draw, minimum size = 0.4cm},
        column sep = {0.5cm},
        row sep={0.35cm}]
{
             & |(0)|   &          & |(2)|    &          &  \\
             &         &          &           &          & |(3)|   \\
|(4)|        &         &          &           &          &          & |(-4)| \\
             &         &|(5)|     &           & |(6)|    &        &   &                & |(00)|   &          & |(02)|    &          &  \\
             &         &          &|(7)|      &          &        &  &                 &         &          &           &          & |(03)|   \\ 
             &         &          &           &          &        &  &  |(04)|        &         &          &           &          &          & |(-04)| \\
             &             & |(1)|  &          & |(i)|    &          &   &             &         &|(05)|     &           & |(06)|    &  \\
&             &         &          &           &          & |(j)|    &                &         &          &|(07)|      &          &   \\   
&|(-k)|       &         &          &           &          &          & |(k)| \\
&             &         &|(-j)|    &           & |(-1)|  &          &\\
&             &         &          &|(-i)|   &          &          &\\ 
};}

\begin{scope}[style=thick]
 \foreach \from/\to/\weight/\where 
                    in { 4/2/1/above,  4/3/1/above, 4/-4/1/right, 4/7/1/above, 4/5/1/right, 
                         0/5/1/above, 0/7/1/above, 0/6/1/above, 0/3/1/above, 0/2/1/above,
                         2/7/1/above, 2/6/1/right, 2/-4/1/above, 
                         3/5/1/right, 3/6/1/below, 3/-4/1/right, -4/7/1/above, -4/5/1/above, 
                         6/7/1/below, 6/5/1/below,
                        -k/i/1/above, -k/j/1/above, -k/k/1/right, -k/-i/1/above, -k/-j/1/right, 
                         1/-j/1/above, 1/-i/1/above, 1/-1/1/above, 1/j/1/above, 1/i/1/above,
                         i/-i/1/above, i/-1/1/right, i/k/1/above, 
                         j/-j/1/right, j/-1/1/below, j/k/1/right, k/-i/1/above, k/-j/1/above, 
                        -1/-i/1/below, -1/-j/1/below,
                         04/02/1/above,  04/03/1/above, 04/-04/1/right, 04/07/1/above, 04/05/1/right, 
                         00/05/1/above, 00/07/1/above, 00/06/1/above, 00/03/1/above, 00/02/1/above,
                         02/07/1/above, 02/06/1/right, 02/-04/1/above, 
                         03/05/1/right, 03/06/1/below, 03/-04/1/right, -04/07/1/above, -04/05/1/above, 
                         06/07/1/below, 06/05/1/below}
    \draw (\from) to [->]    (\to);
\end{scope}
\begin{scope}[style=thin]
 \foreach \from/\to/\weight/\where 
          in { 4/-k/1/above,  0/1/1/above, 2/i/1/right, 3/j/1/above, -4/k/1/right, 
               6/-1/1/above,  5/-j/1/above, 7/-i/1/above}
 \draw[gray] (\from) to [->]    (\to);
          \end{scope}
\begin{scope}[style=thin]
 \foreach \from/\to/\weight/\where 
          in { 04/-k/1/above,  00/1/1/above, 02/i/1/right, 03/j/1/above, -04/k/1/right, 
               06/-1/1/above,  05/-j/1/above, 07/-i/1/above}
  \draw[red] (\from) to [->]    (\to);
          \end{scope}
 \begin{scope}[style=thin]
 \foreach \from/\to/\weight/\where          
            in{ 4/04/1/above,  0/00/1/above, 2/02/1/right, 3/03/1/above, -4/-04/1/right, 
               6/06/1/above,  5/05/1/above, 7/07/1/above}
 \draw[blue] (\from) to [->]    (\to);
          \end{scope}

\end{tikzpicture}
\caption{$Cay(Q_8, \{-1,\pm i, \pm j\})\Box K_3$}
 \end{center}
\end{figure}

\section{Approximation algorithm for covering array on graph}\label{Approx}
In this section, we present an approximation algorithm for construction of covering array on a given graph $G=(V,E)$ with 
$k>1$ prime factors with respect to the Cartesian product.  
In 1988,  G. Seroussi and N H. Bshouty proved that the decision problem whether there exists a binary 
covering array of strength $t\geq 2$ and size $2^t$ on a given $t$-uniform hypergraph is NP-complete \cite{VS}.  
Also, construction of 
an optimal size covering array on a graph is at least as hard as finding its optimal size. 
 
\noindent We give an approximation algorithm  for  the Cartesian  product with  approximation ratio $O(\log_s |V|)$, where $s$  can be obtained from the 
number of symbols corresponding to each vertex. The following result by Bush is used in our approximation algorithm.   

\begin{theorem}\rm{\cite{GT}}\label{B} Let $g$ be a positive integer. If $g$ is written in standard form: $$g=p_1^{n_1}p_2^{n_2}\ldots p_l^{n_l}$$ where $p_1,p_2,\ldots,p_l$ are distinct primes, and if 
$$r=\mbox{min}(p_1^{n_1},p_2^{n_2},\ldots, p_l^{n_l}),$$  then one can construct  $OA(s,g)$ where 
 $s =1+ \max{(2,r)}$.
\end{theorem}
We are given a wighted connected graph $G=(V,E)$ with each vertex having the same weight $g$. 
In our approximation algorithm, we use  a technique from \cite{HBGP} for prime factorization of $G$ with respect to the Cartesian product. 
 This can be done in $O(E \log V$) time. For details see \cite{HBGP}. After obtaining prime factors of $G$, we construct 
 strength two  covering array $C_1$ on maximum size prime factor. Then 
 using  rows of $C_1$,  we produce a covering array on $G$.\\

\noindent\textbf{APPROX $CA(G,g)$:}
\\\textbf{Input:}  A weighted connected  graph $G=(V,E)$ with $k>1$ prime factors with respect to the  Cartesian product. Each vertex has weight $g$; $g=p_1^{n_1}p_2^{n_2}\ldots p_l^{n_l}$ where 
$p_1$, $p_2, \ldots, p_l$ are primes. 
 \\\textbf{Output:} $CA(ug^2,G,g)$.
\\\textbf{Step 1:} Compute $s = 1 + \mbox{max}\{2,r\}$ where $r=\mbox{min}(p_1^{n_1},p_2^{n_2},\ldots, p_l^{n_l})$.
\\\textbf{Step 2:} Factorize $G$ into  prime factors with respect to the Cartesian product;
say $G = \Box_{i=1} ^{k} G_i$ where $G_i= (V_i,E_i)$ is a prime factor.
\\\textbf{Step 3:}  Suppose $V_1\geq V_2\geq \ldots\geq V_k$. For prime factor $G_1=(V_1, E_1)$ \textbf{do} 
\begin{enumerate}
\item Find the smallest positive integer $u$ such that $s^u\geq V_1$. That is, $u=\lceil \mbox{log}_s V_1\rceil$. 
\item Let $OA(s,g)$ be an orthogonal array and denote its $i$th row by $R_i$ for $i=1,2,\ldots,s$.  Total $s^u$ many row vectors $(R_{i_1}, R_{i_2},\ldots R_{i_u})$, each of length $ug^2$,  are formed by horizontally concatenating $u$ rows  
$R_{i_1}$, $ R_{i_2}$, $\ldots,$ $ R_{i_u}$   where $1\leq i_1, \ldots, i_u\leq s$. 
\item Form an $V_1 \times ug^2$ array $C_1$ by choosing any $V_1$ rows  out of $s^u$ concatenated row vectors. 
Each row in the array corresponds to a vertex in the graph $G_1$.  \end{enumerate}
\textbf{Step 4:}
From  $C_1$ we can construct an $V\times ug^2$ array $C$.  Index the rows of $C$ by $(u_1,u_2,\ldots,u_k)$, $u_i\in V(G_i)$. 
Set the row $(u_1,u_2,\ldots,u_k)$ to be identical to the row corresponding to $u_1+u_2+\ldots+u_k ~ \mbox{mod } V_1$   in $C_1$. Return $C$.

\vspace{1cm}\begin{theorem}
  Algorithm APPROX $CA(G,g)$  is a  polynomial-time $\rho(V)$ approximation algorithm for covering array on graph problem,  where 
 $$\rho(V) \leq \lceil \log_s \frac{V}{2^{k-1}} \rceil.$$
  \end{theorem}
\begin{proof}
\textbf{Correctness:} The verification that $C$ is a $CA(ug^2,G,g)$  is straightforward.  First, we show that $C_1$ is a covering array of strength two with $ |V_1|$ parameters. 
Pick any two distinct rows of $C_1$ and consider the sub matrix induced by these two rows. In the sub matrix, there must be a column $(R_i, R_j)^T$ where $i \neq j$. 
Hence each ordered pair of values appears at least  once. 
 Now to show that $C$ is a covering array on $G$, it is sufficient to show that the rows in $C$ for any pair of adjacent vertices $u=(u_1,u_2,\ldots,u_k)$ and $v=(v_1,v_2,\ldots,v_k)$ in $G$ will be qualitatively 
 independent.    We know $u$ and $v$ are adjacent if and only if  $(a_i,b_i)\in E(G_i)$ for exactly one  index $1\leq i\leq k$ and 
 $a_j=b_j$ for $j\neq i$. 
  Hence $ u_1+u_2+ \ldots+u_k   \neq  v_1+v_2+\ldots+v_k ~ \mbox{mod }  V_1$ and  in Step 6, 
 two distinct rows from $C_1$ are assigned  to the vertices $u$ and $v$.\\
 \textbf{Complexity :} The average order of $l$ in Step 1 is $\ln\ln g$ \cite{Riesel}. Thus, the time to find $s$ in Step 1 is $O(\ln \ln g)$.  
 The time to factorize graph $G=(V,E)$ in Step 2 is $O(E \log V)$. In Step 3(1), the smallest positive integer $u$ can be found in 
 $O(\mbox{log}_s V_1)$ time. In Step 3(2), forming one row vector requires $\mbox{log}_sV_1$ assignments; hence, forming $V_1$ row vectors require $O(V_1\mbox{log}V_1)$ time. 
 Thus the total running time of APPROX $CA(G,g)$ is  $O(E \log V+\ln \ln g)$. Observing that, in practice, $\ln \ln g \leq E \log V$, we can restate the running time of 
 APPROX $CA(G,g)$ as $O(E \log V)$.  \\
 \textbf{Approximation ratio:}  We show that APPROX $CA(G,g)$ returns a covering array that is at most $\rho(V)$ times the size of an optimal covering array on $G$. 
 We know the smallest $n$ for which a $CA(n,G,g)$ exists is $g^2$, that is, $CAN(G,g)\geq g^2$. The algorithm returns a covering array on $G$ of size $ug^2$ where
 $$u=\lceil \log_s V_1\rceil.$$ As $G$ has $k$ prime factors, the maximum number of vertices in a factor can be $\frac{V}{2^{k-1}}$, that is, $V_1\leq \frac{V}{2^{k-1}}$.  
Hence $$u= \lceil \log_s V_1\rceil \leq \lceil \log_s \frac{V}{2^{k-1}}\rceil.$$ By relating to the size of the covering array returned to the optimal size, we obtain our approximation ratio 
$$\rho(V)\leq \lceil \log_s \frac{V}{2^{k-1}}\rceil.$$  \end{proof}

\section{Conclusions}  One motivation for introducing  a graph structure was to optimise covering arrays for their use in testing software and networks based on internal structure. Our primary 
concern in this paper is with constructions that make optimal covering arrays on large graphs from smaller ones. Large graphs are obtained by considering either the Cartesian, the direct, the strong, or the Lexicographic product of small graphs.  Using graph homomorphisms, we have 
$$\max_{i=1,2}\{CAN(G_i,g)\}\leq CAN(G_1\Box G_2,g)\leq CAN( \max_{i=1,2}\{\chi(G_i)\},g).$$ We gave several classes of Cayley graphs where the lower bound on covering array number $CAN(G_1\Box G_2)$ is achieved. It is an interesting problem to find out other classes of graphs for which lower bound on covering array number of product graph can be achieved. We gave an approximation algorithm 
for construction of covering array on a graph $G$ having more than one factor with respect to the Cartesian product. Clearly, another area to explore is to consider in details the other graph products, that is, the direct, the strong, and the Lexicographic product.

\end{document}